\documentclass[10pt,conference]{IEEEtran}
\IEEEoverridecommandlockouts

\usepackage[utf8]{inputenc} 
\usepackage[T1]{fontenc}
\usepackage{url}
\usepackage{ifthen}
\usepackage{cite}
\usepackage{algorithmicx}
\usepackage{algorithm}
\usepackage{graphicx}
\usepackage{textcomp}
\usepackage{xcolor}
\usepackage{multirow}
\usepackage{multicol}
\usepackage{caption}
\usepackage{subcaption}
\usepackage{amsthm}
\usepackage{amssymb}
\usepackage{svg}
\usepackage{stfloats}
\usepackage[noend]{algpseudocode}

\usepackage[cmex10]{amsmath} % Use the [cmex10] option to ensure complicance
\theoremstyle{definition}
\newtheorem{theorem}{Theorem}

\allowdisplaybreaks

\DeclareMathSizes{10}{9}{6.5}{5} 
\begin{document}
    \title{Hyper-VIB: A Hypernetwork-Enhanced Information Bottleneck Approach for Task-Oriented Communications}

\author{\IEEEauthorblockN{
Jingchen Peng\IEEEauthorrefmark{1}, 
Chaowen Deng\IEEEauthorrefmark{1}, 
Yili Deng\IEEEauthorrefmark{1}, 
Boxiang Ren\IEEEauthorrefmark{1} and
Lu Yang\IEEEauthorrefmark{2}$^\ddag$}
\IEEEauthorblockA{\IEEEauthorrefmark{1}Department of Mathematical Sciences, Tsinghua University, Beijing, China}
\IEEEauthorblockA{\IEEEauthorrefmark{2}Wireless Technology Lab, Central Research Institute, 2012 Labs, Huawei Tech. Co. Ltd., China}
  \IEEEauthorblockA{Email: yanglu87@huawei.com}
  \vspace{-2em}
\thanks{The first two authors contributed equally to this work and $\ddag$ marked the corresponding author. }
 }

% \author{\IEEEauthorblockN{
% Jingchen Peng\IEEEauthorrefmark{1}, 
% Boxiang Ren\IEEEauthorrefmark{1}, 
% Lu Yang\IEEEauthorrefmark{2}$^\ddag$, 
% Chenghui Peng\IEEEauthorrefmark{2},
% Panpan Niu\IEEEauthorrefmark{1},
% and Hao Wu\IEEEauthorrefmark{1}}
% \IEEEauthorblockA{\IEEEauthorrefmark{1}Department of Mathematical Sciences, Tsinghua University, Beijing, China}
% \IEEEauthorblockA{\IEEEauthorrefmark{2}Wireless Technology Lab, Central Research Institute, 2012 Labs, Huawei Tech. Co. Ltd., China}
%   \IEEEauthorblockA{Email: yanglu87@huawei.com}
% \thanks{The first two authors contributed equally to this work and $\ddag$ marked the corresponding author. }
%  }

\maketitle

%%%%%%
%% Abstract: 
%% If your paper is eligible for the student paper award, please add
%% the comment "THIS PAPER IS ELIGIBLE FOR THE STUDENT PAPER
%% AWARD." as a first line in the abstract. 
%% For the final version of the accepted paper, please do not forget
%% to remove this comment!
%%

\begin{abstract}
%This paper addresses the optimization of task-oriented data transmission within collaborative intelligence frameworks, tackling a critical challenge in sixth-generation (6G) mobile networks. Grounded in classical Information Bottleneck (IB) theory, our approach minimizes communication overhead while maximizing AI task accuracy by integrating feature extraction with channel transmission. Although the performance metric derived from IB theory provides a natural optimization objective, its direct computation becomes intractable for high-dimensional data. To overcome this limitation, we derive a variational upper-bound approximation to reformulate and solve the problem. To further improve optimization efficiency, we introduce Hyper-VIB, a hypernetwork-enhanced information bottleneck approach for task-oriented communications, which incorporates a hypernetwork module that converges to optimal hyperparameter configurations within a single end-to-end training session. Empirical validation demonstrates that our algorithm outperforms current state-of-the-art techniques in both effectiveness and efficiency.

This paper presents Hyper-VIB, a hypernetwork-enhanced information bottleneck (IB) approach designed to enable efficient task-oriented communications in 6G collaborative intelligent systems. Leveraging IB theory, our approach enables an optimal end-to-end joint training of device and network models, in terms of the maximal task execution accuracy as well as the minimal communication overhead, through optimizing the trade-off hyperparameter. To address computational intractability in high-dimensional IB optimization, a tractable variational upper-bound approximation is derived. Unlike conventional grid or random search methods that require multiple training rounds with substantial computational costs, Hyper-VIB introduces a hypernetwork that generates approximately optimal DNN parameters for different values of the hyperparameter within a single training phase. Theoretical analysis in the linear case validates the hypernetwork design. Experimental results demonstrate our Hyper-VIB's superior accuracy and training efficiency over conventional VIB approaches in both classification and regression tasks.
\end{abstract}

\begin{IEEEkeywords}
Task-oriented communications, hyperparameter optimization, hypernetwork, information bottleneck, variational inference 
\end{IEEEkeywords}

%\vspace{-0.9em}
\section{Introduction}
The accelerated evolution of AI has propelled its adoption across critical domains, including autonomous systems, immersive technologies (e.g. AR/VR), and holographic communication etc. These widespread integrations, identified by the International Telecommunication Union (ITU) as core deployment scenarios for emerging 6G wireless systems \cite{recommendation2023framework}, have triggered unprecedented growth in data transmission requirements, fundamentally reshaping traditional communication paradigms. Consequently, conventional data-oriented communication frameworks \cite{yang2019data}, historically focused on complete data reconstruction and intrinsic signal properties, are giving way to task-oriented methodologies \cite{shao2021learning, shi2023task, xie2023robust}. These emerging approaches focus on optimizing transmission for specific objectives like real-time inference and distributed intelligence, emphasizing efficiency and relevance over exhaustive data delivery.\par
%This paradigm evolution underscores the imperative for future networks to natively unify task-driven communication, distributed computation, and AI orchestration.
Despite significant advancements, the absence of a comprehensive theoretical framework for analyzing trade-offs between communication costs and inference performance remains a key obstacle in distributed AI. Information bottleneck (IB) theory \cite{tishby2000information} provides a promising foundation for quantifying information loss and transmission efficiency. Recent studies \cite{shao2021learning, xie2023robust} have applied IB to communication system design, aiming to balance overhead with task accuracy. Extensions such as distributed IB (DIB) theory \cite{moldoveanu2021network, moldoveanu2023network} adapt IB to multi-device scenarios \cite{shao2022task}. However, these approaches often overlook channel noise and, critically, fail to enable collaborative training between device and network sides.\par
While recent work \cite{10619157} has addressed both channel noise and collaborative training concerns, training collaborative intelligent systems requires determining hyperparameters that balance performance and cost. Although current methods, including grid search \cite{10619157}, random search, and Bayesian optimization \cite{snoek2015scalable}, perform well in low-dimensional settings, they become computationally prohibitive in higher dimensions due to the independent training for each hyperparameter configuration. Hypernetworks, introduced by \cite{ha2016hypernetworks}, are neural networks designed to generate weights for a target network and offer a promising alternative for hyperparameter tuning. Recent studies \cite{bae2023multirate, mackayself} leverage hypernetworks to approximate optimal response functions \cite{gibbons1992primer}, mapping hyperparameters directly to optimized model parameters. In communication systems, hypernetworks have also been applied to characterize the varying channel states in the encoder-decoder design \cite{xie2024deep}, with higher adaptability and efficiency. However, integrating the hypernetwork into IB-based collaborative intelligent system design remains unexplored.

In this paper, we introduce Hyper-VIB, a hypernetwork-enhanced framework grounded in IB theory for task-oriented communications. This approach not only enables rapid and efficient identification of optimal hyperparameters, overcoming the large computation overhead and time delay in conventional methods, but also dynamically adapts to varying channel conditions, with its effectiveness rigorously proved through theoretical analysis. Our main contributions are as follows: 1) The first scheme that combines hypernetworks with IB theory in communication systems, establishing an intelligent system for accelerated hyperparameter optimization. The system can thus adaptively select the optimal configurations based on task requirements and channel states, enabling efficient task-oriented communications; 2) Theoretical validation of the Hyper-VIB framework through optimality analysis of linear models. We establish the existence of a closed-form optimal response function that validates the hypernetwork design, providing foundational theoretical guarantees for Hyper-VIB's efficacy; 3) Development of the Hyper-IB framework with support for different kinds of intelligent tasks, including the classification tasks, and the regression tasks. Extensive experiments on image classification and wireless localization demonstrate that Hyper-VIB significantly accelerates model training while maintaining competitive performance, highlighting its broad applicability and superior performance.

%In this paper, we introduce Hyper-VIB, a hypernetwork-enhanced framework grounded in IB theory for task-oriented communications.  This approach not only enables rapid and efficient identification of optimal hyperparameters, overcoming the computational overhead and time limitations of conventional methods, but also dynamically adapts to varying channel conditions, with its effectiveness rigorously validated through theoretical analysis.  Our core contributions are summarized as follows: 1) First integration of hypernetworks with IB theory and communication technology, establishing an intelligent system for accelerated hyperparameter optimization.  This system dynamically determines task-optimal configurations based on task requirements and channel states, achieving highly efficient task-oriented communications;  2) Development of the Hyper-IB framework with dual-task capability (classification and regression).  Extensive experiments on image classification (classification task) and wireless localization (regression task) demonstrate that Hyper-IB drastically reduces optimization time while maintaining competitive task performance, highlighting its broad applicability.

\section{System Model and Problem Formulation}
This section introduces the device-network collaborative intelligence system model, illustrated in Fig.~\ref{fig: system model}. In addition, our optimization problem based on the IB theory for a single link of device-network collaboration is formulated \cite{shao2021learning}.
\begin{figure}[t]
  \centering
  \includegraphics[width=1\linewidth]{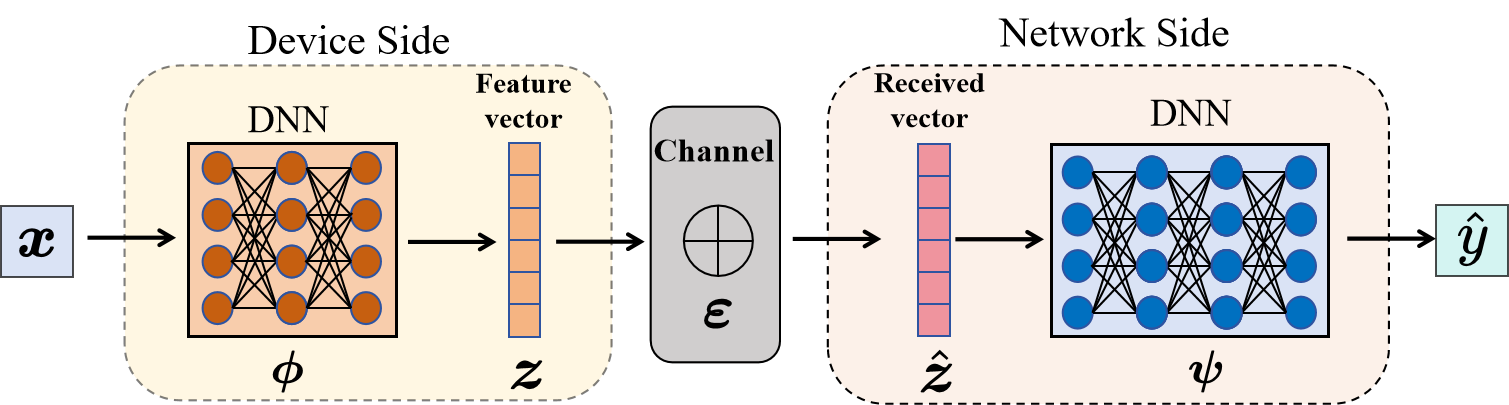}
  \caption{The system model of device-network collaborative intelligence.}
  \label{fig: system model}
  \vspace{-1.5em}
\end{figure}

\subsection{System Model and Data Transmission Chains}\label{subsec: system model}
The input data $\boldsymbol{x}$ at the mobile device and the target variable $\boldsymbol{y}$ (e.g., a label) are considered as the realizations of the random variables $(X, Y)$, characterized by the joint distribution $p(\boldsymbol{x}, \boldsymbol{y})$. Given the dataset $\{\boldsymbol{x}^{(m)}\}_{m=1}^M$, the device and network in Fig. \ref{fig: system model} should collaboratively infer the targets $\{\boldsymbol{y}^{(m)}\}_{m=1}^M$. 

The device, upon receiving $\boldsymbol{x}$, extracts and encodes a feature representation $\boldsymbol{z}$. We model this combined feature extraction and encoding process as a probabilistic encoder $p_{\boldsymbol{\phi}}(\boldsymbol{z}|\boldsymbol{x})$, where $\boldsymbol{\phi}$ represents the trainable DNN parameters on the device. To establish a well-defined probabilistic model for the device, we utilize the reparameterization trick \cite{kingma2015variational}. Then the conditional distribution $p_{\boldsymbol{\phi}}(\boldsymbol{z}|\boldsymbol{x})$ is modeled as a multivariate independent Gaussian distribution, that is, 
\begin{equation}\label{eq: conditional x_k to z_k}
p_{\boldsymbol{\phi}}(\boldsymbol{z}|\boldsymbol{x}) = \mathcal{N}(\boldsymbol{z}|\boldsymbol{\mu},\boldsymbol{\Theta}).
\end{equation}
Here, the mean vector $\boldsymbol{\mu} \in \mathbb{R}^d$ is determined by the DNN-based function $\boldsymbol{\mu}(\boldsymbol{x},\boldsymbol{\phi})$ and the covariance matrix $\boldsymbol{\Theta}$ is structured as a diagonal matrix $\mathrm{diag} \{\theta_{1}^2,\dots,\theta_{d}^2\}$, where $\boldsymbol{\theta} = (\theta_{1},\dots, \theta_{d})$ is produced by the DNN-based function $\boldsymbol{\theta}(\boldsymbol{x},\boldsymbol{\phi})$. 

Subsequently, the encoded feature $\boldsymbol{z}$ undergoes the wireless channel to the BS (i.e., network side), obtaining the received feature $\hat{\boldsymbol{z}}$. This transmission channel, characterized by the conditional distribution $p(\tilde{\boldsymbol{z}}|\boldsymbol{z})$, is modeled as a non-trainable layer. Without loss of generality, we consider an additive white Gaussian noise channel, resulting in $\hat{\boldsymbol{z}} = \boldsymbol{z} + \boldsymbol{\varepsilon}$, where $\boldsymbol{\varepsilon} \sim \mathcal{N}(\boldsymbol{0},\sigma^2 \boldsymbol{I})$ denotes the noise with variance $\sigma^2$.

Finally, the network side performs inference using the received features $\hat{\boldsymbol{z}}$. This process produces the inference output $\hat{\boldsymbol{y}}$ via a DNN-based function $\hat{\boldsymbol{y}}(\hat{\boldsymbol{z}},\boldsymbol{\psi})$, where $\boldsymbol{\psi}$ denotes the parameter set of the network-side DNN.

The above processes construct a Markov chain:
\begin{equation}\label{eq: Markov chain}
    X \leftrightarrow Z \leftrightarrow \hat{Z} \leftrightarrow \hat{Y},
\end{equation}
where the random variables $Z$, $\hat{Z}$ and $\hat{Y}$ correspond to the encoded feature $\boldsymbol{z}$, the received feature $\hat{\boldsymbol{z}}$ and the inference result $\hat{\boldsymbol{y}}$, respectively. With this formulation, the probability equation is given by $p(\hat{\boldsymbol{y}},\hat{\boldsymbol{z}},\boldsymbol{z}|\boldsymbol{x}) =p(\hat{\boldsymbol{y}}|\hat{\boldsymbol{z}})
p(\hat{\boldsymbol{z}}|\boldsymbol{z}) p_{\boldsymbol{\phi}} (\boldsymbol{z}|\boldsymbol{x})$.

\subsection{Information Bottleneck (IB) Theory}
To achieve high accuracy of collaborative AI tasks with minimal wireless transmission overhead, we aim to develop the optimal AI models for both the device and network sides. According to \cite{shao2021learning}, conventional IB theory provides an effective framework for the single-link collaboration. Specifically, consider the following performance metric
\begin{equation}\label{eq: IB}
\mathcal{C}_{\text{IB}}=-I(Y;\hat{Z}) + \beta I(X;\hat{Z}),
\end{equation}
where the objective function constitutes a weighted sum of two mutual information terms, with hyperparameter $\beta > 0$ governing the trade-off and to be discussed in detail in section \ref{sec: hypernetwork}. Here, the first mutual information $I(Y;\hat{Z})$ quantifies the information about target variable $Y$ preserved in the network-side received features $\hat{Z}$, serving as a representation of inference accuracy. Meanwhile, the second term $I(X;\hat{Z})$ evaluates the information retained in $\hat{Z}$ conditioned on $X$ through the minimum description length principle \cite{cover1999elements}, thereby characterizing communication overhead. %Furthermore, this term captures the collective impact of encoding and transmission processes on the raw input data. 
Consequently, a higher AI task accuracy and lower communication overhead are signified by smaller values of \(\mathcal{C}_{\text{IB}}\), yielding the following formulation of the optimization problem
\begin{equation}\label{eq: problem}
    \vspace{-0.5em}
    \min_{\boldsymbol{\phi}, \boldsymbol{\psi}} \,\,
    \mathcal{C}_{\text{IB}},
\end{equation}
which is with the goal of identifying optimal AI models for device side (i.e., ${\boldsymbol{\phi}}$) and network side (i.e., $\boldsymbol{\psi}$) 
%to collectively minimize the performance metric $\mathcal{C}_{\text{ML-IB}}$.

\subsection{Variational Upper Bound of $\mathcal{C}_{\text{IB}}$}
According to the definition of mutual information, calculating $\mathcal{C}_{\text{IB}}$ in \eqref{eq: IB} needs the conditional distribution $p(\boldsymbol{y}|\hat{\boldsymbol{z}})$ for $I(Y;\hat{Z})$ and marginal distributions ${p(\hat{\boldsymbol{z}})}$ for ${I(X;\hat{Z})}$. Computational challenges arise, however, from the inherent high dimensionality of input data, which impedes accurate distribution computation. We therefore leverage variational approximation \cite{kingma2013auto, alemi2016deep}, a widely-adopted technique for approximating intractable distributions using adjustable parameters (e.g., DNN weights). Specifically, variational distribution ${r(\hat{\boldsymbol{z}})}$ is introduced to approximate ${p(\hat{\boldsymbol{z}})}$, which is modeled as a centered isotropic Gaussian distribution $\mathcal{N}(\hat{\boldsymbol{z}}|\boldsymbol{0},\boldsymbol{I})$ \cite{alemi2016deep}. Since the inference variable $\hat{Y}$ and the target variable $Y$ share identical value ranges, the variational conditional distribution $p_{\boldsymbol{\psi}}(\boldsymbol{y}|\hat{\boldsymbol{z}})$ can be approximated by $d\left(\boldsymbol{y}, \hat{\boldsymbol{y}}\left(\hat{\boldsymbol{z}}, \boldsymbol{\psi}\right)\right)$, where $d(\boldsymbol{y},\hat{\boldsymbol{y}})$ is the loss function to measure the distortion between $\boldsymbol{y}$ and $\hat{\boldsymbol{y}}$, such as cross-entropy or mean squared error (MSE).
% Since the inference variable $\hat{Y}$ and the target variable $Y$ share identical value ranges, the variational conditional distribution is established as
% \begin{equation}
% p_{\boldsymbol{\psi}}(\boldsymbol{y}|\hat{\boldsymbol{z}}) \propto \exp\left(-d\left(\boldsymbol{y}, \hat{\boldsymbol{y}}\left(\hat{\boldsymbol{z}}, \boldsymbol{\psi}\right)\right)\right),
% \end{equation}
% where $d(\boldsymbol{y},\hat{\boldsymbol{y}})$ is the loss function to measure the distortion between $\boldsymbol{y}$ and $\hat{\boldsymbol{y}}$, such as cross-entropy or mean squared error (MSE).

Based on the variational approximation method \cite{kingma2013auto, alemi2016deep}, an upper bound for $\mathcal{C}_{\text{IB}}$ can be derived using lower and upper bounds of its mutual information terms $I(Y;\hat{Z})$ and $I(X; \hat{Z})$, respectively. To be specific, the non-negativity property of KL divergence, i.e. $D_{\text{KL}}(p(\boldsymbol{y}|\hat{\boldsymbol{z}}) \| p_{\boldsymbol{\psi}}(\boldsymbol{y}|\hat{\boldsymbol{z}})) \geq 0$, implies the existence of  a lower bound for $I(Y;\hat{Z})$ as
\begin{equation}\label{eq: upbound1}
    I(Y;\hat{Z}) 
    \geq \int p(\hat{\boldsymbol{z}},\boldsymbol{y})\log \frac{p_{\boldsymbol{\psi}}(\boldsymbol{y}|\hat{\boldsymbol{z}})}{p(\boldsymbol{y})} \, d\hat{\boldsymbol{z}} \, d\boldsymbol{y}.
\end{equation}
Similarly, the non-negativity of \( D_{\text{KL}}(p(\hat{\boldsymbol{z}}) \| r(\hat{\boldsymbol{z}})) \) enables derivation of an upper bound for $I(X; \hat{Z})$ given by
\begin{equation}\label{eq: upbound2}
    I(X; \hat{Z}) \leq \int p(\hat{\boldsymbol{z}},\boldsymbol{x})\log\frac{p_{\boldsymbol{\phi}}(\hat{\boldsymbol{z}}|\boldsymbol{x})}{r(\hat{\boldsymbol{z}})}d\hat{\boldsymbol{z}} d\boldsymbol{x}.
\end{equation}

By combining \eqref{eq: upbound1} and \eqref{eq: upbound2}, we can derive an upper bound of $\mathcal{C}_{\text{IB}}$ ignoring a constant term of $H(Y)$ as
\begin{equation}\label{eq: VML-IB}
\begin{aligned}
    \mathcal{C}_{\text{VIB}}(\boldsymbol{\phi},\boldsymbol{\psi}) &= \mathbb{E}_{p(\boldsymbol{x},\boldsymbol{y})} \Big \{ \mathbb{E}_{p_{\boldsymbol{\phi}}(\hat{\boldsymbol{z}}|\boldsymbol{x})}[-\log p_{\boldsymbol{\psi}}(\boldsymbol{y}|\hat{\boldsymbol{z}})] \\
\quad &\quad\quad\quad\quad\quad\quad\quad\quad+ \beta D_{\text{KL}}(p_{\boldsymbol{\phi}}(\hat{\boldsymbol{z}}|\boldsymbol{x}) \| r(\hat{\boldsymbol{z}}))  \Big \}.
\end{aligned}
\end{equation}
% \begin{equation}\label{eq: VML-IB}
%     \mathcal{C}_{\text{VIB}} \!=\! \mathbb{E}_{p(\boldsymbol{x},\boldsymbol{y})} \Big \{ \mathbb{E}_{p_{\boldsymbol{\phi}}(\hat{\boldsymbol{z}}|\boldsymbol{x})}[-\!\log p_{\boldsymbol{\psi}}(\boldsymbol{y}|\hat{\boldsymbol{z}})] \! +\!  \beta D_{\text{KL}}(p_{\boldsymbol{\phi}}(\hat{\boldsymbol{z}}|\boldsymbol{x}) \| r(\hat{\boldsymbol{z}}))  \Big \}.
% \end{equation}
Note that $\hat{\boldsymbol{z}}=\boldsymbol{z} + \boldsymbol{\varepsilon}$ with $\boldsymbol{\varepsilon} \sim \mathcal{N}(\boldsymbol{0},\sigma^2 \boldsymbol{I})$ and $\boldsymbol{\varepsilon}$ is independent of $\boldsymbol{x}$ and $\boldsymbol{z}$, 
%thus $\hat{\boldsymbol{z}}$, being the sum of two Gaussian distributions, also follows a Gaussian distribution 
according to \eqref{eq: conditional x_k to z_k}, the conditional probability density function for $\hat{\boldsymbol{z}}$ is expressed as
\begin{equation}\label{eq: conditional distribution}
p_{\boldsymbol{\phi}}(\hat{\boldsymbol{z}}|\boldsymbol{x}) = \mathcal{N}(\hat{\boldsymbol{z}}|\boldsymbol{\mu},\boldsymbol{\Theta}+\sigma^2\boldsymbol{I}).
\end{equation}
Therefore, the KL divergence $D_{\text{KL}}(p_{\boldsymbol{\phi}}(\hat{\boldsymbol{z}}|\boldsymbol{x}) \| r(\hat{\boldsymbol{z}}))$ is computable through its definition as
% \begin{equation}\label{eq: KL divergence}
% D_{\text{KL}}(p_{\boldsymbol{\phi}}(\hat{\boldsymbol{z}}|\boldsymbol{x}) \| r(\hat{\boldsymbol{z}})) = \dfrac{1}{2} \left( \|\boldsymbol{\mu}\|^{2} + \sum_{i=1}^{d}(\theta_{i}^{2} + \sigma^{2}) - \sum_{i=1}^{d}\log(\theta_{i}^{2} + \sigma^{2}) - d \right)
% \end{equation}
\begin{equation}\label{eq: KL divergence}
\begin{aligned}
    D_{\text{KL}}(p_{\boldsymbol{\phi}}(\hat{\boldsymbol{z}}|\boldsymbol{x}) \| r(\hat{\boldsymbol{z}})) &= \\\dfrac{1}{2} \left( \|\boldsymbol{\mu}\|^{2} + \sum_{i=1}^{d}\right.&\left.(\theta_{i}^{2} + \sigma^{2}) - \sum_{i=1}^{d}\log(\theta_{i}^{2} + \sigma^{2}) - d \right).
\end{aligned}
\end{equation}

\section{The Hyper-VIB Approach}\label{sec: hypernetwork}
% From the IB objective in \eqref{eq: IB}, it is evident that setting the hyperparameter $\beta$ is an indispensable issue for achieving optimal performance. 
% Traditional methods such as grid search and random search requires fixing $\beta$ and retraining the model for each search, which significantly increases both computational cost and time. 

% % 解决方法
% In this section, we present the Hyper-VIB approach, which efficiently generates model parameters for different $\beta$ values, enabling evaluation and selection of the optimal $\beta$ in a single training run.
% Particularly, we formulate a best-response function that maps $\beta$ to the optimal model parameters and explicitly model this function with a hypernetwork. 
% Furthermore, we justify the validity of the hypernetwork architecture by demonstrating that, under a linear model, it can represent the optimal response at the device end.
The IB objective in \eqref{eq: IB} highlights the critical role of the hyperparameter $\beta$ in achieving the optimal performance.
Traditional methods, including the grid search and the random search, require retraining the model for each fixed $\beta$, leading to high computational cost and intolerable time delay.

% 解决方法
To address this, we propose Hyper-VIB in this section, which enables efficient determination of optimal model parameters with diverse values of $\beta$ within a single training run.
Particularly, we formulate a best-response function that maps $\beta$ to the optimal model parameters and explicitly model this function with a hypernetwork. 
Furthermore, we justify the validity of the hypernetwork architecture by demonstrating that, under a linear model, it can represent the optimal response at the device end.
\vspace{-0.5em}
\subsection{Hypernetwork for Best Response}
We begin by constructing the best-response function \cite{gibbons1992primer}
\begin{equation}\label{eq: response}
\begin{aligned}
\boldsymbol{\phi}^*(\beta), \boldsymbol{\psi}^*(\beta)
= \arg \min_{\boldsymbol{\phi}, \boldsymbol{\psi}} \, \, \mathcal{C}_{\text{VIB}}(\beta).
\end{aligned}
\end{equation}
This function maps the hyperparameter $\beta$ to the corresponding optimal DNN parameters: $\boldsymbol{\phi}^*(\beta)$ for the device side and $\boldsymbol{\psi}^*(\beta)$ for the network side, both trained with $\beta$.\par
Next, the hypernetwork is employed to model the best-response function, which is a single neural network designed to generate parameters for a family of target networks, offering advantages such as weight sharing, dynamic tuning, and high efficiency \cite{ha2016hypernetworks, xie2024deep}.
In this paper, $\boldsymbol{\phi}^*(\beta)$ and $\boldsymbol{\psi}^*(\beta)$ are regarded as the target networks. Inspired by \cite{bae2023multirate}, we construct efficient hypernetworks while ensuring no significant increase in storage requirements or training time.
\begin{figure}[t]
% \raggedright
\includegraphics[width=\linewidth]{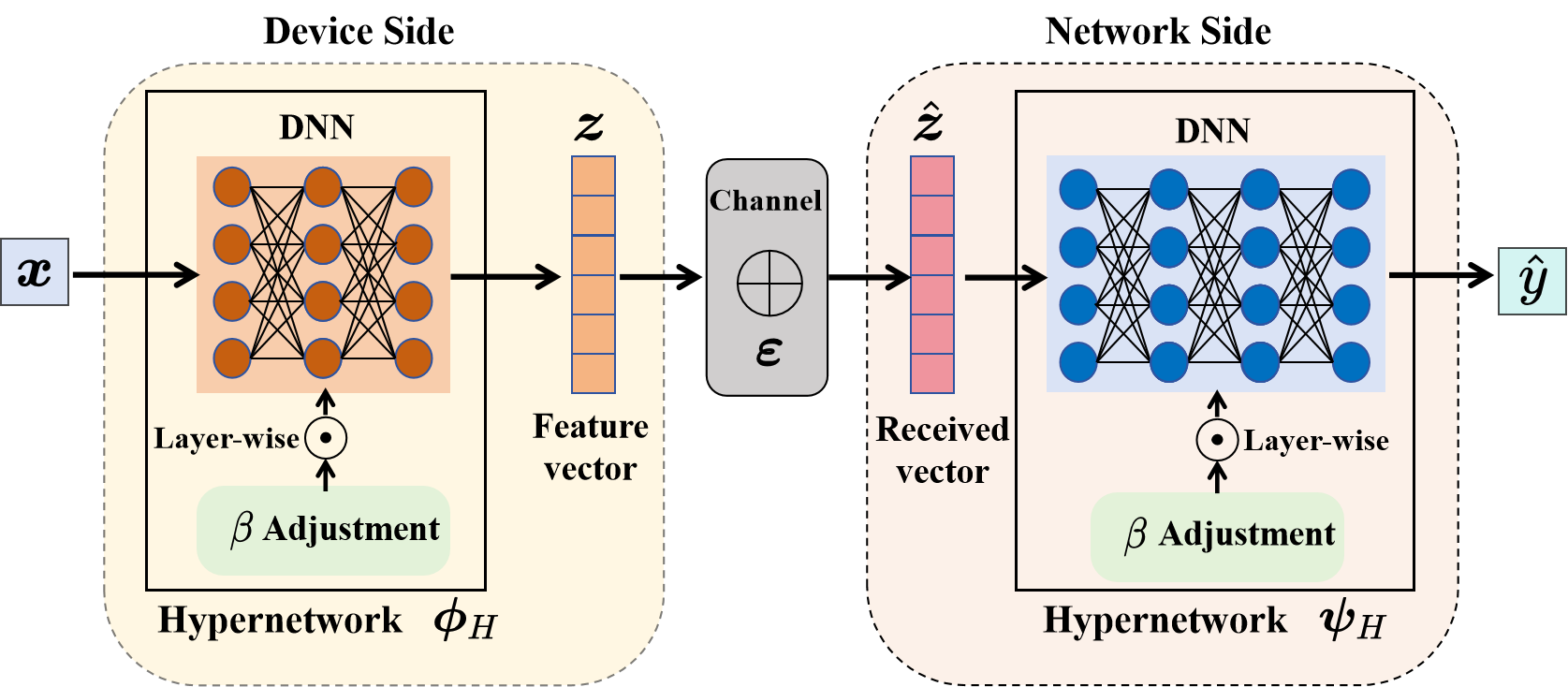}
% \caption{Collaborative intelligence system with quantization}
\caption{The hypernetwork framework with DNN and $\beta$ adjustment module in Hyper-VIB.}
\label{fig: architecture for hypernetwork}
\vspace{-1.5em}
\end{figure}

As shown in Fig.\ \ref{fig: architecture for hypernetwork}, our hypernetwork consists of two main components: the original DNN and the $\beta$ adjustment module, which introduces the influence of the hyperparameter $\beta$ at each network layer.
Together, these components form two hypernetworks: $\boldsymbol{\phi}_{H}$ on the device side and $\boldsymbol{\psi}_{H}$ on the network side. When $\beta$ is fixed, these hypernetworks generate the corresponding parameters, $\boldsymbol{\phi}_{H}(\beta)$ and $\boldsymbol{\psi}_{H}(\beta)$, which approximate the best-response function in \eqref{eq: response}.
\par
Then, we focus on how the $\beta$ adjustment module applies layer-wise adjustments on the DNNs in our hypernetwork. Specifically, we consider both fully connected and convolutional layers, as described below.
\subsubsection{\textbf{Fully connected layers}} Let the input feature to the $i$-th layer be $\boldsymbol{d}^{(i-1)} \in \mathbb{R}^{D_{i-1}}$, and the output feature be $\boldsymbol{d}^{(i)} \in \mathbb{R}^{D_{i}}$. Denote the weights for this layer as $\boldsymbol{W}^{(i)} \in \mathbb{R}^{D_{i} \times D_{i-1}}$ and the biases as $\boldsymbol{b}^{(i)} \in \mathbb{R}^{D_{i}}$. The output is computed as
\begin{equation}
\boldsymbol{d}^{(i)} = \boldsymbol{W}^{(i)} \boldsymbol{d}^{(i-1)} + \boldsymbol{b}^{(i)}.
\end{equation}
With the $\beta$ adjustment module now introduced, the parameters for this layer are adjusted as follows: 
\begin{equation}\label{eq: layer-wise-1}
    \begin{aligned}
        \boldsymbol{W}^{(i)}(\beta) &= S^{(i)}\left(\boldsymbol{u}^{(i)}\log(\beta) + \boldsymbol{v}^{(i)} \right)\odot_{\text{row}} \boldsymbol{W}^{(i)}, \\
         \boldsymbol{b}^{(i)}(\beta) &= S^{(i)}\left(\boldsymbol{u}^{(i)}\log(\beta) + \boldsymbol{v}^{(i)} \right)\odot  \boldsymbol{b}^{(i)}.
    \end{aligned}
\end{equation}
Here, $S^{(i)}$ is the activation function, in particular, the sigmoid function. 
The symbols $\odot$ and $\odot_{\text{row}}$ denote element-wise and row-wise multiplication, respectively. In \eqref{eq: layer-wise-1}, $\boldsymbol{u}^{(i)}$, $\boldsymbol{v}^{(i)}$ $\in \mathbb{R}^{D_i}$ perform an affine transformation on $\log(\beta)$ to obtain scaling factors, with $S^{(i)}$ ensuring that the scaling remains within the range $[0,1]$.
Among these, $\boldsymbol{W}^{(i)}$ and $\boldsymbol{b}^{(i)}$ are the original DNN parameters, while $\boldsymbol{u}^{(i)}$ and $\boldsymbol{v}^{(i)}$ in $\mathbb{R}^{D_{i}}$ are the $\beta$ adjustment module parameters, together forming the hypernetwork parameters at the $i$-th layer. This configuration shows that, compared to the original model, our hypernetwork introduces only $2D_{i}$ additional parameters at this layer, demonstrating its compactness and memory efficiency. Therefore, the output $\boldsymbol{d}^{(i)}$ can be expressed as
\begin{equation}
\boldsymbol{d}^{(i)} \!=\! S^{(i)} \!\left(\boldsymbol{u}^{(i)} \log(\beta) \!+\! \boldsymbol{v}^{(i)}\right) \!\odot \!\left(\boldsymbol{W}^{(i)} \boldsymbol{d}^{(i-1)} \!+\! \boldsymbol{b}^{(i)}\right).
\end{equation}
The layer-wise hypernetwork design is depicted in Fig. \ref{fig: layer-wise architecture for hypernetwork}.

\begin{figure}[t]
% \raggedright
\includegraphics[width=\linewidth]{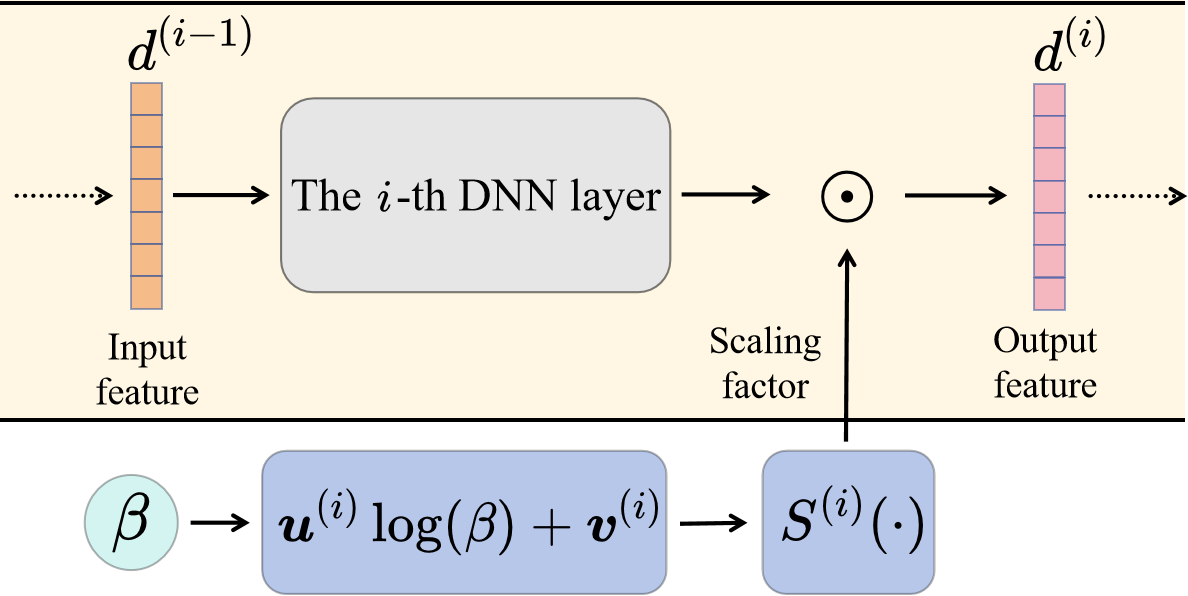}
% \caption{Collaborative intelligence system with quantization}
\caption{Layer-wise architecture of the hypernetwork.}
\label{fig: layer-wise architecture for hypernetwork}
\vspace{-1.5em}
\end{figure}

\subsubsection{\textbf{Convolutional layers}} Consider the $(i-1)$-th layer with $C_{i-1}$ filters, the $i$-th layer with $C_i$ filters and kernel size $K_i$. Let $\boldsymbol{W}^{(i,c)} \in \mathbb{R}^{C_{i-1} \times K_i \times K_i}$ be the weight of the $c$-th filter and $b^{(i,c)}$ be its bias. Like fully connected layers, now introducing the $\beta$ adjustment module, the parameters are adjusted as 
\begin{equation}\label{eq: layer-wise-3}
    \begin{aligned}
        \!\!\!\!\boldsymbol{W}^{(i,c)}(\beta) &\!=\! S^{(i,c)}\!\!\left(\boldsymbol{u}_1^{(i,c)}\!\log(\beta) \!+\! \boldsymbol{v}_1^{(i,c)} \right)\!\odot_{\text{kernel}}\! \boldsymbol{W}^{(i,c)}, \\
         b^{(i)}(\beta) &\!=\! S^{(i,c)}\!\!\left(u_2^{(i,c)}\!\log(\beta) \!+\! v_2^{(i,c)} \right)\!\odot\!  b^{(i,c)}.
    \end{aligned}
\end{equation}
Here, $\boldsymbol{u}_1^{(i,c)}$, $\boldsymbol{v}_1^{(i,c)} \in \mathbb{R}^{C_{i-1}}$, and $u_2^{(i,c)}$, $v_2^{(i,c)} \in \mathbb{R}$. The symbol $\odot_{\text{kernel}}$ denotes element-wise multiplication applied to the $K_i\times K_i$ blocks of $\boldsymbol{W}^{(i,c)}$. Together, the set $\{\boldsymbol{W}^{(i,c)}, \boldsymbol{b}^{(i,c)}, \boldsymbol{u}_1^{(i,c)}, \boldsymbol{v}_1^{(i,c)}, u_2^{(i,c)}, v_2^{(i,c)}\}_{c=1}^{C_i}$ forms the hypernetwork parameters at the $i$-th layer. Compared to the original model, our hypernetwork introduces only $2C_i(C_{i-1}+1)$ additional parameters at this layer.

%\vspace{-0.5em}
\subsection{Hyper-VIB objective and training algorithm design}
To learn the optimal hypernetworks parameters $\boldsymbol{\phi}_{H}$ and $\boldsymbol{\psi}_{H}$, our Hyper-VIB approach adopts the following objective function:
\begin{equation}\label{eq: hyper-VIB}
    \mathcal{C}_{\text{Hyper-VIB}}(\boldsymbol{\phi}_{H},\boldsymbol{\psi}_{H}) =
    \mathbb{E}_{\beta \sim U[a,b]} \  \mathcal{C}_{\text{VIB}}\left( \boldsymbol{\phi}_{H}(\beta),\boldsymbol{\psi}_{H}(\beta)\right),
\end{equation}
where the hyperparameter $\beta$ follows the uniform distribution between $[a,b]$. Our optimization problem can be formulated as
\begin{equation}
    \vspace{-0.5em}
    \boldsymbol{\phi}^{*}_{H}, \boldsymbol{\psi}^{*}_{H} = \arg \min_{\boldsymbol{\phi}_{H},\boldsymbol{\psi}_{H}} \mathcal{C}_{\text{Hyper-VIB}}(\boldsymbol{\phi}_{H},\boldsymbol{\psi}_{H}).
\end{equation}
This objective guides the hypernetworks to learn the optimal response parameters $\boldsymbol{\phi}_{H}(\beta)$ and $\boldsymbol{\psi}_{H}(\beta)$ for each fixed $\beta$ within the range $[a,b]$, with the goal of minimizing $\mathcal{C}_{\text{VIB}}$. \par
Using Monte Carlo sampling, we can obtain an unbiased estimation of the gradient for $\mathcal{C}_{\text{Hyper-VIB}}$. 
Specifically, during training, the hyperparameter $\beta$ is sampled $T$ times, and for each sampled $\beta^{(t)}$, given a mini-batch of data $\{\boldsymbol{x}^{(t,m)}\}_{m=1}^M$, the empirical estimation is acquired:
\begin{equation}\label{eq: simeq}
\hspace{-1em}
\begin{aligned}
        \mathcal{C}_{\text{Hyper-VIB}} &\simeq \frac{1}{T} \sum_{t=1}^T \frac{1}{M} \sum_{m=1}^M \left\{ -\log p_{\boldsymbol{\psi}_{H}(\beta^{(t)})}(\boldsymbol{y}^{(t,m)}|\hat{\boldsymbol{z}}^{(t,m)}) \right. \\
        &\!\!\!\! \left. + \beta^{(t)} D_{\text{KL}}\left( p_{\boldsymbol{\phi}_{H}(\beta^{(t)})}(\hat{\boldsymbol{z}}^{(t,m)}|\boldsymbol{x}^{(t,m)}) \middle\| r(\hat{\boldsymbol{z}}^{(t,m)}) \right) \right\}.
    \end{aligned}
\end{equation}
Then, we propose an algorithm to learn the optimal hypernetwork parameters, as summarized in Algorithm~\ref{al: 1}.
\begin{figure}[t]
\vspace{-1em}
\begin{algorithm}[H]
\caption{Training Procedure for Hyper-VIB}  
\label{al: 1} 
\begin{algorithmic}[1] 
    \renewcommand{\algorithmicrequire}{\textbf{Input:}}
    \renewcommand{\algorithmicensure}{\textbf{Output:}}
    \Require {Training dataset $\mathcal{D}$, channel noise variance $\sigma^2$, range $[a,b]$ for $\beta$.}
    \Ensure {Optimized hypernetwork parameters $\boldsymbol{\phi}_{H}, \boldsymbol{\psi}_{H}$.}
    \While{not converged}
        \For{$t = 1$ \textbf{to} $T$}
            \State Sample $\beta^{(t)} \sim U[a,b]$
            \State Select a mini-batch of data $\{(\boldsymbol{x}^{(t,m)}, y^{(t,m)})\}_{m=1}^M$
            
            \For{$m = 1$ \textbf{to} $M$}
                \State Compute feature vector $\boldsymbol{z}^{(t,m)}$ using \eqref{eq: conditional x_k to z_k}
                \State Sample noise $\{\boldsymbol{\varepsilon}^{(l,m)}\}_{l=1}^L \sim \mathcal{N}(\boldsymbol{0}, \sigma^2 \boldsymbol{I})$
                \State Compute received vector $\hat{\boldsymbol{z}}^{(m)} = \boldsymbol{z}^{(m)} + \boldsymbol{\varepsilon}^{(l,m)}$
            \EndFor
        \EndFor
        \State Compute KL-divergence using \eqref{eq: KL divergence}
        \State Compute loss $\mathcal{C}_{\text{Hyper-VIB}}$ using \eqref{eq: simeq}
        \State Update parameters via backpropagation
    \EndWhile
\end{algorithmic}  
\end{algorithm}
\vspace{-2.5em}
\end{figure}
%\vspace{-0.5em}
\subsection{Optimality Analysis in the Linear Case}
To validate the design of our hypernetwork architecture, we analyze the optimal response function under linear model assumptions, which have been widely adopted in prior studies such as \cite{bae2023multirate, mackayself}.
Specifically, we consider the input data $\boldsymbol{x} \in \mathbb{R}^n$ and the target output $y \in \mathbb{R}$. On the device side, we consider a linear network $\boldsymbol{A} \in \mathbb{R}^{d \times n}$, which uses deterministic reparameterization. On the network side, we consider a linear network $\boldsymbol{B} \in \mathbb{R}^{d \times 1}$. The data transmission flow is
%\vspace{-0.5em}
\begin{equation}\label{eq: liner-model}
    %\vspace{-0.5em}
    \boldsymbol{z} = \boldsymbol{A}\boldsymbol{x},\ \hat{\boldsymbol{z}} = \boldsymbol{z}+\boldsymbol{\varepsilon},\  \hat{y} =   \boldsymbol{B}^T \hat{\boldsymbol{z}}.
\end{equation}
Here, $\boldsymbol{\varepsilon} \sim \mathcal{N}(0, \sigma^2 \boldsymbol{I})$, and we choose the MSE $d = (y - \hat{y})^2$ to measure the distortion between $y$ and $\hat{y}$.
%\vspace{-1em}
\begin{theorem}
%For the linear model above, when the network-side hypernetwork parameters $\boldsymbol{\psi}_{H}(\beta)$ are given, there always exist device-side parameters $\boldsymbol{\phi}_{H}(\beta)$ in our hypernetwork that are exactly equal to the optimal response function $\boldsymbol{\phi}(\beta)$ from \eqref{eq: response} for any $\beta$.
\vspace{-0.5em}
Consider the linear model described in \eqref{eq: liner-model}. Given the network-side hypernetwork parameters $\boldsymbol{\psi}_{H}(\beta)$, there exist device-side parameters $\boldsymbol{\phi}_{H}(\beta)$ in our hypernetwork such that $\boldsymbol{\phi}_{H}(\beta) = \boldsymbol{\phi}^*(\beta)$ for all $\beta$, where $\boldsymbol{\phi}^*(\beta)$ is the optimal response function defined in \eqref{eq: response}.
\vspace{-0.5em}
\end{theorem}

\begin{proof}
Let the training set be $\mathcal{D} = \{\boldsymbol{x}^{(i)}, y^{(i)}\}_{i=1}^N$, organized by a matrix $\boldsymbol{X} = [\boldsymbol{x}^{(1)}, \dots, \boldsymbol{x}^{(N)}] \in \mathbb{R}^{n \times N}$ and a vector $\boldsymbol{y} = [y^{(1)}, \dots, y^{(N)}] \in \mathbb{R}^{1 \times N}$, with $N \gg n$. With MSE, \eqref{eq: VML-IB} yields
\vspace{-0.5em}
\begin{equation}
\vspace{-0.5em}
\begin{aligned}
    \mathcal{C}_{\text{VIB}}(\boldsymbol{\phi}, \boldsymbol{\psi}) 
    &= \frac{1}{N} \sum_{i=1}^N \left\{ \mathbb{E}_{p_{\boldsymbol{\phi}}(\hat{\boldsymbol{z}}^{(i)} | \boldsymbol{x}^{(i)})} \left[ (\hat{y}^{(i)} - y^{(i)})^2 \right] \right. \\
    &\quad \quad \quad \quad \left. + \beta D_{\text{KL}} \left( p_{\boldsymbol{\phi}}(\hat{\boldsymbol{z}}^{(i)} | \boldsymbol{x}^{(i)}) \| p(\hat{\boldsymbol{z}}^{(i)}) \right) \right\},\end{aligned}
\end{equation}
%Here, $\boldsymbol{\phi}$ and $\boldsymbol{\psi}$ represent the network parameters on the device and network sides, respectively. 
For the linear model in \eqref{eq: liner-model}, the device-side parameters $\boldsymbol{\phi}_{H}(\beta)$ and network-side parameters $\boldsymbol{\psi}_{H}(\beta)$ correspond to $\boldsymbol{A}$ and $\boldsymbol{B}$, respectively.
%In the following derivations, we replace them with $\boldsymbol{A}$ and $\boldsymbol{B}$.\par
Based on this model, we have $\hat{\boldsymbol{z}}^{(i)} = \boldsymbol{A} \boldsymbol{x}^{(i)} + \boldsymbol{\varepsilon}^{(i)}$, implying that $p_{\boldsymbol{\phi}}(\hat{\boldsymbol{z}}^{(i)} | \boldsymbol{x}^{(i)}) = \mathcal{N}(\boldsymbol{A} \boldsymbol{x}^{(i)}, \sigma^2 \boldsymbol{I})$. Thus,
\begin{equation}\label{eq: proof-1}
\begin{aligned}
\mathbb{E}_{p_{\boldsymbol{\phi}}(\hat{\boldsymbol{z}}^{(i)} | \boldsymbol{x}^{(i)})} (\hat{y}^{(i)} - y^{(i)})^2 &=  \mathbb{E}[\boldsymbol{B}^T (\boldsymbol{A}\boldsymbol{x}^{(i)}+\boldsymbol{\varepsilon}^{(i)}) - y^{(i)}]^2 \\
&=(\boldsymbol{B}^T \boldsymbol{A}\boldsymbol{x}_i - y_i)^2 + \sigma^2 \boldsymbol{B}^T \boldsymbol{B}.
\end{aligned}
\vspace{-0.5em}
\end{equation}
From \eqref{eq: KL divergence}, we have
\begin{equation}\label{eq: proof-2}
\begin{aligned}
&D_{\text{KL}} ( p_{\boldsymbol{\phi}}(\hat{\boldsymbol{z}}^{(i)} | \boldsymbol{x}^{(i)}) \| p(\hat{\boldsymbol{z}}^{(i)})\\
 & \quad\quad\quad\quad  = \frac{1}{2} \left\{(\boldsymbol{A} \boldsymbol{x}_i)^T (\boldsymbol{A} \boldsymbol{x}_i) +  d(\sigma^2 -  \log \sigma^2 - 1) \right\}.
  \end{aligned}  
\end{equation}
Combining \eqref{eq: proof-1} and \eqref{eq: proof-2}, $\mathcal{C}_{\text{VIB}}(\boldsymbol{A}, \boldsymbol{B})$ can be derived as
\begin{equation}\label{eq: proof-3}
\begin{aligned}
    \mathcal{C}_{\text{VIB}}(\boldsymbol{A}, \boldsymbol{B}) 
    = \frac{1}{N}& \left\{  ||\boldsymbol{y} - \boldsymbol{B}^T \boldsymbol{A} \boldsymbol{X}||^2 +  \frac{\beta}{2} \text{tr}(\boldsymbol{X}^T \boldsymbol{A}^T \boldsymbol{A} \boldsymbol{X})  \right\}\\
    & + \sigma^2 \boldsymbol{B}^T \boldsymbol{B}+ \frac{\beta}{2} \left( \sigma^2 d - d \log \sigma^2 - d \right).
\end{aligned}
\end{equation}
It is evident that $\mathcal{C}_{\text{VIB}}(\boldsymbol{A}, \boldsymbol{B})$ in \eqref{eq: proof-3} is a convex function with respect to $\boldsymbol{A}$. The gradient with respect to $\boldsymbol{A}$ is
\begin{equation}\label{eq: proof-4}
   \frac{\partial\mathcal{C}_{\text{VIB}}(\boldsymbol{A}, \boldsymbol{B})}{\partial \boldsymbol{A}} \!\!=\!\!  \frac{1}{N}\!\!\left\{\!-2 \boldsymbol{B}\boldsymbol{y}X^T \!\!\!+\! 2\boldsymbol{B}\boldsymbol{B}^T\!\!\boldsymbol{A}\boldsymbol{X}\boldsymbol{X}^T \!\!\!+\! \beta \boldsymbol{A}\boldsymbol{X}\boldsymbol{X}^T\!\right\}\!.
   %\vspace{-1em}
\end{equation}
Since $\boldsymbol{X} \in \mathbb{R}^{n \times N}$ and $N \gg n$, $\boldsymbol{X}\boldsymbol{X}^T$ is invertible by assuming $\text{Rank}(\boldsymbol{X}) = n$. Setting the gradient in \eqref{eq: proof-4} to zero yields
\begin{equation}
    \boldsymbol{A}^* = \left(\boldsymbol{B}\boldsymbol{B}^T+ \frac{1}{2}\beta \boldsymbol{I} \right)^{-1} \boldsymbol{B}\boldsymbol{y}\boldsymbol{X}^T (\boldsymbol{X}\boldsymbol{X}^T)^{-1},
\end{equation}
which is the closed-form solution for the optimal response function on the device side with fixed $\boldsymbol{B}$. Clearly, $\boldsymbol{A}^*$ depends on $\beta$ and varies with $\beta$. 

We now prove that there exist parameters in our hypernetwork that can represent $\boldsymbol{A}^*$.
Consider the singular value decomposition of $\boldsymbol{B} = \boldsymbol{U} \boldsymbol{\Sigma} V^T$, where $\boldsymbol{U} \in \mathbb{R}^{d \times d}$ is a unitary orthogonal matrix, $\boldsymbol{\Sigma} \in \mathbb{R}^{d \times 1}$ is a vector with only the first element equal to $||\boldsymbol{B}||$, and the rest are zeros, while $V$ is essentially a scalar 1. Then $\boldsymbol{A}^*$ can be expressed as
\begin{equation}
    \label{eq:A*_SVD}
    \boldsymbol{A}^* = \boldsymbol{U}\left(\boldsymbol{\Sigma} \boldsymbol{\Sigma}^T +\frac{1}{2}\beta \boldsymbol{I}\right)^{-1}\boldsymbol{\Sigma} \boldsymbol{y}\boldsymbol{X}^T (\boldsymbol{X}\boldsymbol{X}^T)^{-1}.
\end{equation}\par
Consider the device-side network as a two-layer structure with
%$\boldsymbol{A} = \boldsymbol{W}^{(2)} \boldsymbol{W}^{(1)}$. Let 
$\boldsymbol{W}^{(2)} = 2\boldsymbol{U}$ and $\boldsymbol{W}^{(1)} = \frac{1}{||\boldsymbol{B}||^2}\boldsymbol{\Sigma} \boldsymbol{y} \boldsymbol{X}^T (\boldsymbol{X}\boldsymbol{X}^T)^{-1}$. For the first-layer $\beta$ adjustment module, let $\boldsymbol{u}^{(1)} = -1 \cdot \mathbf{1}$ and $\boldsymbol{v}^{(1)} = (\ln 2 + 2 \ln ||B||) \cdot \mathbf{1}$, where $\mathbf{1} \in \mathbb{R}^{d}$ is the all-ones vector. Therefore, we have
%\vspace{-0.5em}
\begin{equation}
S^{(1)}(\boldsymbol{u}^{(1)}\log(\beta)\!+\!\boldsymbol{v}^{(1)}) \!=\! 
\frac{1}{1\!+\!\frac{1}{2}{||\boldsymbol{B}||}^{-2}\beta} \cdot \mathbf{1}\!=\! \frac{||\boldsymbol{B}||^2}{||\boldsymbol{B}||^2\!+\!\frac{1}{2}\beta} \cdot \mathbf{1}.
\end{equation}
Continuing the derivation, we get
\begin{equation}
\label{eq:W1_beta}
\begin{aligned}
    \boldsymbol{W}^{(1)}(\beta) &= S^{(1)}(\boldsymbol{u}^{(1)}\log(\beta) + \boldsymbol{v}^{(1)}) \odot_{\text{row}} \boldsymbol{W}^{(1)} \\
    &= \text{diag} \!\left( \!\frac{||\boldsymbol{B}||^2}{||\boldsymbol{B}||^2 + \frac{1}{2}\beta} \!\right)\! \cdot \frac{1}{||\boldsymbol{B}||^2} \boldsymbol{\Sigma} \boldsymbol{y} \boldsymbol{X}^T (\boldsymbol{X}\boldsymbol{X}^T)^{-1} \\
    &= \left(\boldsymbol{\Sigma} \boldsymbol{\Sigma}^T + \frac{1}{2}\beta \right)^{-1} \boldsymbol{\Sigma} \boldsymbol{y} \boldsymbol{X}^T (\boldsymbol{X}\boldsymbol{X}^T)^{-1},
\end{aligned}
\end{equation}
where \(\text{diag}(x)\) represents a diagonal matrix in \(\mathbb{R}^{d \times d}\) with all diagonal elements are $x$. For the second layer, setting both \(\boldsymbol{u}^{(2)}\) and \(\boldsymbol{v}^{(2)}\) to zero vectors yields
\begin{equation}
    \label{eq:W2_beta}
    \boldsymbol{W}^{(2)}(\beta) = S^{(2)}(\boldsymbol{u}^{(2)}\log(\beta) + \boldsymbol{v}^{(2)}) \odot_{\text{row}} \boldsymbol{W}^{(2)} = \boldsymbol{U}.
\end{equation}
% \begin{equation}
% \begin{aligned}
%     \boldsymbol{W}^{(2)}(\beta) &= S^{(2)}(\boldsymbol{u}^{(2)}\log(\beta) + \boldsymbol{v}^{(2)}) \odot_{\text{row}} \boldsymbol{W}^{(2)} \\
%     &= \text{diag} \left( \frac{1}{2} \right) \cdot 2\boldsymbol{U} = \boldsymbol{U}.
% \end{aligned}
% \end{equation}
According to \eqref{eq:A*_SVD}, \eqref{eq:W1_beta} and \eqref{eq:W2_beta}, we have 
\begin{equation}
    \boldsymbol{A}^* = \boldsymbol{W}^{(2)}(\beta) \boldsymbol{W}^{(1)}(\beta).
\end{equation}
Finally, by setting $\boldsymbol{\phi}_{H} (\beta) = \boldsymbol{W}^{(2)}(\beta) \boldsymbol{W}^{(1)}(\beta)$, we conclude that for any $\beta$, when $\boldsymbol{\psi}_{H}(\beta) = \boldsymbol{B}$ is given, $\boldsymbol{\phi}_{H} (\beta)$ is the optimal response function for $\beta$.
\end{proof}

\section{Numerical Results and Discussions}
This section evaluates the performance of the proposed Hyper-VIB approach in device-network collaborative intelligence, highlighting its efficiency in enabling task-oriented communications.

\subsection{Experimental Settings}
We consider classification and regression tasks. For classification, we use the MNIST \cite{lecun1998gradient} and CIFAR-10 \cite{krizhevsky2009learning} datasets. For regression, we study wireless localization in multipath environments using simulated Bluetooth and WiFi ranging data. The simulated datasets for received signals are generated based on the Bluetooth two-way channel model \cite{tariq2024reduced} and the WiFi channel model \cite{kotaru2015spotfi}, with the target output being the estimated line-of-sight (LoS) propagation time between transceivers.
To assess the impact of the hypernetwork, we compare Hyper-VIB with the classical VIB method \cite{shao2021learning,qian2021variational}. For a fair comparison, both Hyper-VIB and VIB use identical DNN architectures, including fully connected layers and convolutional neural networks.
For Hyper-VIB, a single training run is conducted with $\beta$ values ranging from $10^{-5}$ to $1$. For VIB, a grid search is performed, where training is conducted separately for each $\beta$ to determine the optimal value.
We use top-1 accuracy for classification performance and mean squared error (MSE) of propagation time for localization accuracy.
The signal-to-noise ratio (SNR) is set to 20 dB, and the feature dimension $\boldsymbol{z}$ is fixed at 64.

\begin{table}[]
    \vspace{0.5em}
    \centering
    \begin{tabular}{c c c c c}
    \hline
     \textbf{Tasks} & \textbf{Dataset}  & \textbf{Methods}  & \textbf{Parameters} & \textbf{Time(s)} \\
     \hline
     \multirow{4}{*}{Classification}& \multirow{2}{*}{MNIST} & VIB & $18.81 \times 10^4$ & 3,103\\
     % \cline{2-4}
   &   & Hyper-VIB & $18.94 \times 10^4$ &  696 \\ 
   \cline{2-5}
    &    \multirow{2}{*}{CIFAR-10} & VIB & $40.46 \times 10^5$ &   236,800\\
     % \cline{2-4}
   &   & Hyper-VIB & $41.81 \times 10^5$ &  31,640\\
      \hline
   \multirow{4}{*}{Regression}& \multirow{2}{*}{Bluetooth} & VIB & $55.92 \times 10^4$ & 1,759\\
     % \cline{2-4}
   &   & Hyper-VIB & $56.21 \times 10^4$ & 479 \\ 
   \cline{2-5}
    &    \multirow{2}{*}{WIFI} & VIB & $65.42 \times 10^4$ &   2,042\\
     % \cline{2-4}
   &   & Hyper-VIB & $66.48 \times 10^4$ &  584\\
      \hline
      
    \end{tabular}
    \caption{Comparison of parameter counts and training time between VIB and proposed Hyper-VIB.}
    \label{tab:my_label}
    \vspace{-1.8em}
\end{table}

\subsection{Experimental Results}
% \begin{table*}[]
%     \centering
%     \begin{tabular}{c c c c c}
%     \hline
%      \textbf{Tasks} & \textbf{Dataset}  & \textbf{Methods}  & \textbf{Parameters} & \textbf{Training time(s)} \\
%      \hline
%      \multirow{4}{*}{Classification}& \multirow{2}{*}{MNIST} & VIB & $18.81 \times 10^4$ & 3,103\\
%      % \cline{2-4}
%    &   & Hyper-VIB (ours) & $18.94 \times 10^4$ &  696 \\ 
%    \cline{2-5}
%     &    \multirow{2}{*}{CIFAR-10} & VIB & $40.46 \times 10^5$ &   236,800\\
%      % \cline{2-4}
%    &   & Hyper-VIB (ours) & $41.81 \times 10^5$ &  31,640\\
%       \hline
%    \multirow{4}{*}{Regression}& \multirow{2}{*}{Bluetooth Localization} & VIB & $55.92 \times 10^4$ & 1,759\\
%      % \cline{2-4}
%    &   & Hyper-VIB (ours) & $56.21 \times 10^4$ & 479 \\ 
%    \cline{2-5}
%     &    \multirow{2}{*}{WIFI Localization} & VIB & $65.42 \times 10^4$ &   2,042\\
%      % \cline{2-4}
%    &   & Hyper-VIB (ours) & $66.48 \times 10^4$ &  584\\
%       \hline
      
%     \end{tabular}
%     \caption{Comparison of parameter counts and training efficiency between VIB and Hyper-VIB}
%     \label{tab:my_label}
%     \vspace{-2em}
% \end{table*}
In our experiments, to determine the hyperparameter $\beta$, Hyper-VIB is trained once over the range $[10^{-5}, 1]$. After training, we sample $\beta$ every 0.5 orders of magnitude and use the hypernetwork to directly generate the corresponding model parameters for evaluation, selecting the $\beta$ that achieves the best performance. 
In contrast, VIB performs a grid search over the same $\beta$ range, training a separate model for each sampled value, and selects the $\beta$ that yields the best performance. \par
Firstly, we compare the training time and parameter count of Hyper-VIB and VIB across multiple datasets.
As shown in Table~\ref{tab:my_label}, Hyper-VIB achieves a substantial improvement in training efficiency, reducing training time by approximately 80\% compared to VIB. This efficiency gain is attributed to the introduction of a hypernetwork, which allows Hyper-VIB to be trained only once and eliminates the need for retraining across all $\beta$ values, as required by VIB’s grid search.
In terms of parameter count, Hyper-VIB is nearly identical to VIB, as only the $\beta$ adjustment module in the hypernetwork adds a negligible number of additional parameters.\par
Then, we compare the performance of Hyper-VIB and VIB.
As shown in Fig.~\ref{fig:four_results}, across all four datasets, Hyper-VIB consistently outperforms or rivals VIB for all $\beta$ values, and when both methods achieve peak performance respectively, the corresponding trade-off hyperparameter $\beta$ exhibits closely aligned values.
This demonstrates the effectiveness of the hypernetwork in Hyper-VIB at approximating the best-response function from $\beta$ to model parameters.
As a result, Hyper-VIB requires only a single training run to generate model parameters for all $\beta$ values, evaluate their performance, and efficiently identify the optimal one, thereby significantly reducing training time and computational cost.\par
It can be observed from Fig. \ref{fig:four_results} that a smaller $\beta$ improves the accuracy in classification tasks, and reduces the regression errors in localization tasks.
In addition, if $\beta$ is large and far from the optimal region, VIB shows a fast drop in performance.
This is because VIB trains separate models for each fixed $\beta$, making it highly sensitive to hyperparameter values. Instead, Hyper-VIB undergoes a single training run over all values of $\beta$ and shares most model weights through the hypernetwork. As such, it can maintain a relatively better performance with a suboptimal $\beta$.
Therefore, within a wide range of $\beta$ values, Hyper-VIB demonstrates a significant performance gain over VIB, highlighting its superior robustness.

\begin{figure}[t]
    \centering
    \includegraphics[width=0.98\linewidth]{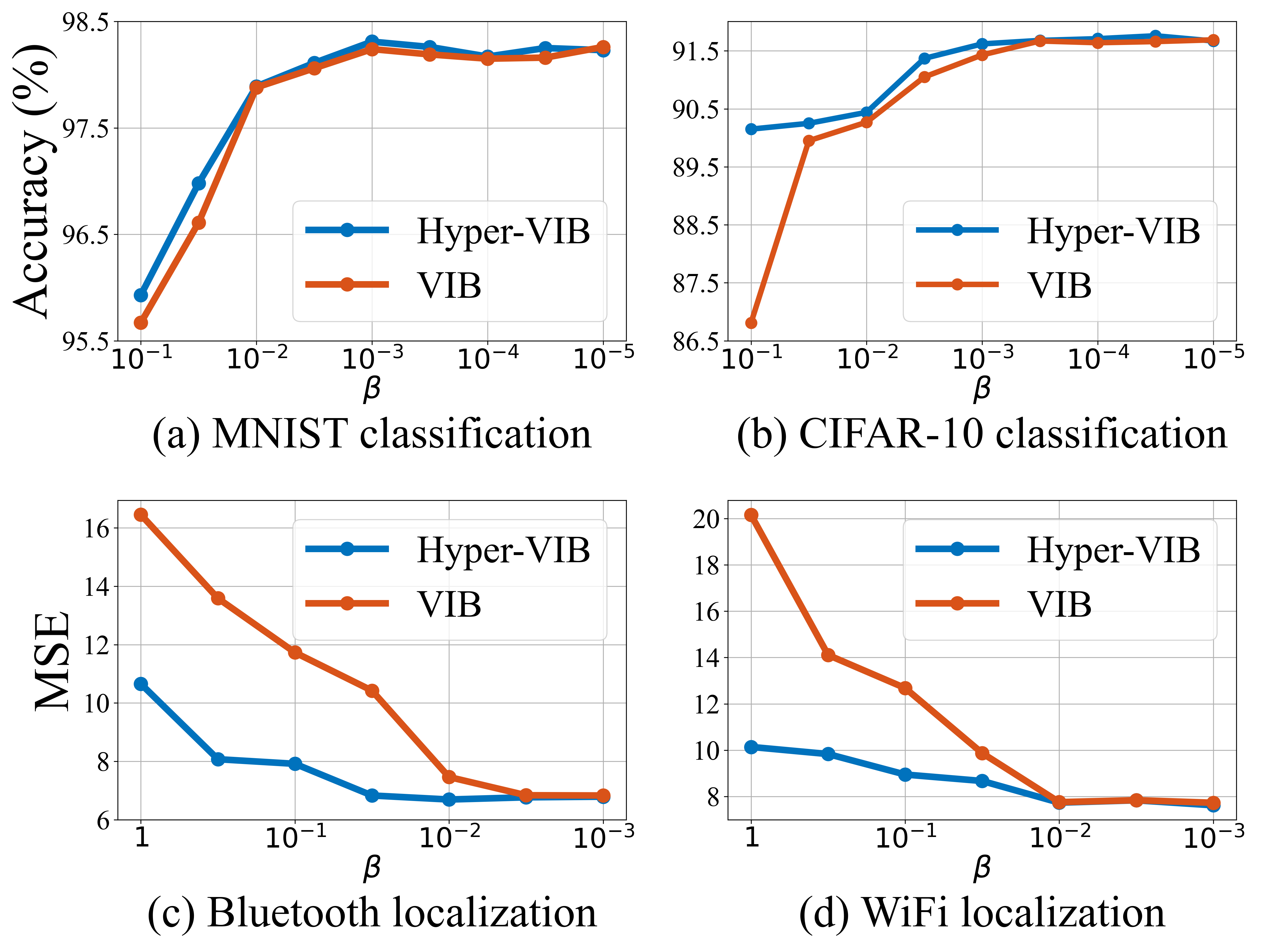}
    \caption{Performance comparisons of our Hyper-VIB scheme with the conventional VIB benchmark, under varying $\beta$ values for different tasks.}
    \label{fig:four_results}
    \vspace{-0.3em}
\end{figure}

\section{Conclusion}
In this paper, we proposed Hyper-VIB, a novel hypernetwork-enhanced information bottleneck framework for efficient task-oriented communications in 6G collaborative intelligent systems. Our approach leverages IB theory to jointly train the device-side and network-side models, which can achieve the minimal communication overhead and maximal task accuracy simultaneously. Hyper-VIB employs hypernetworks to optimize the trade-off hyperparameter $\beta$, which balances task accuracy against communication costs. Specifically, the hypernetwork generates approximately optimal DNN parameters for all $\beta$ values within a single training phase, enabling efficient hyperparameter selection. We further provide theoretical validation through linear case analysis of the hypernetwork design. Experiments verify that Hyper-VIB achieves higher accuracy and training efficiency than classical VIB in both classification and regression tasks.

% In this paper, we introduced Hyper-VIB, a novel hypernetwork-enhanced Information Bottleneck approach addressing the critical challenge of task-oriented communications optimization in 6G collaborative intelligence systems. 

% Unlike conventional grid-search based methods that inefficiently tune hyperparameters, Hyper-VIB enables joint optimization of hyperparameters and DNN parameters within a unified framework, dramatically reducing training time while enhancing efficiency. 

% Our methodology integrates feature extraction with channel transmission through an IB framework, simultaneously minimizing communication overhead and maximizing AI task accuracy. To resolve computational intractability in high-dimensional IB optimization, we derived a tractable variational upper-bound approximation. Empirical results demonstrate that the proposed framework outperforms state-of-the-art methods in both task accuracy and training efficiency, establishing a new paradigm for task-oriented collaborative intelligence.

%\clearpage
\bibliographystyle{IEEEtran}
\bibliography{bibliofile}

\clearpage
\appendices

\end{document}